\documentclass{llncs}
\usepackage{amssymb}
\usepackage{amsmath}
\usepackage{alltt}
\usepackage{enumerate}
\usepackage{tabularx}
\usepackage{theorem}
\usepackage{algorithmic, algorithm}
\usepackage{chngcntr}
\counterwithout{section}{chapter} 
\usepackage{color}

\usepackage[usenames,dvipsnames]{pstricks}
\usepackage{epsfig}
\usepackage{pst-grad} 
\usepackage{pst-plot} 
\usepackage{verbatim}

\usepackage{latexsym}
\usepackage{clrscode3e}
\newtheorem{lema}{Lemma} 

\usepackage{listings}
\usepackage{caption}
\DeclareCaptionFont{white}{\color{white}}
\DeclareCaptionFormat{listing}{\colorbox[cmyk]{0.0, 0.0, 0.0,0.1765}{\parbox{\textwidth}{\hspace{15pt}#1#2#3}}}
\newcommand{\ourtitle}{Identifying Maximal Non-Redundant Integer Cone Generators}

\title{\ourtitle}
\author{Slobodan Mitrovi\'{c} \and Ruzica Piskac \and Viktor Kun\v{c}ak}

\institute{EPFL School of Computer and Communication Sciences,
  Switzerland\footnote{This report is the draft that was submitted to a conference in February 2011 when all authors were at EPFL,
    based on an introductory EPFL graduate program semester
    project by Slobodan Mitrovi\'c.}}

\authorrunning{Slobodan Mitrovi\'{c}, Ruzica Piska\v{c}, Viktor Kuncak}
\titlerunning{\ourtitle}

\newcommand{\smartparagraph}[1]{\smallskip\noindent 
{\bf #1}\ }

\newcommand{\m}[1]{\mbox{\rm \textsf{#1}}}

\newcommand{\justBAPA}{\m{\small BAPA}}
\newcommand{\BAPA}{\m{\small QFBAPA}}

\newcommand{\FPA}{F_{\textsf{PA}}}

\begin{document}
\maketitle
\begin{abstract}
  A non-redundant integer cone generator (NICG) of dimension
  $d$ is a set $S$ of vectors from $\{0,1\}^d$ whose vector
  sum cannot be generated as a positive integer linear
  combination of a proper subset of $S$. The largest
  possible cardinality of NICG of a dimension $d$, denoted
  by $N(d)$, provides an upper bound on the sparsity of
  systems of integer equations with a large number of
  integer variables. A better estimate of $N(d)$ means that
  we can consider smaller sub-systems of integer equations
  when solving systems with many integer
  variables. Furthermore, given that we can reduce
  constraints on set algebra expressions to constraints on
  cardinalities of Venn regions, tighter upper bound on $N(d)$
  yields a more efficient decision procedure for a logic of
  sets with cardinality constraints (BAPA), which has been
  applied in software verification.  Previous attempts to
  compute $N(d)$ using SAT solvers have not succeeded even
  for $d=3$. The only known values were computed manually:
  $N(d)=d$ for $d < 4$ and $N(4) > 4$. We provide the first
  exact values for $d > 3$, namely, $N(4)=5$, $N(5)=7$, and
  $N(6)=9$, which is a significant improvement of the known
  asymptotic bound (which would give only e.g. $N(6) \le
  29$, making a decision procedure impractical for
  $d=6$). We also give lower bounds for $N(7)$, $N(8)$,
  $N(9)$, and $N(10)$, which are: $11$, $13$, $14$, and
  $16$, respectively. We describe increasingly sophisticated
  specialized search algorithms that we used to explore the
  space of non-redundant generators and obtain these
  results.
\end{abstract}

\section{Introduction}\label{sec:introduction}

The theory of sets and set operations plays an important
role in software verification and data flow analysis
\cite{Aiken99SetConstraintsIntro}.  Additionally, reasoning
about sets is used for proving correctness of data
structures, since a natural choice of an abstraction
function is the abstraction function that maps the content
of a data structure to a set. For full functional
verification of complex data structures often it is
important to maintain the number of elements stored in the
data structure
\cite{ZeeETAL08FullFunctionalVerificationofLinkedDataStructures}. The
logic in which one can express set relations, cardinality
constraints and linear integer arithmetic is known under the
name Boolean Algebra with Presburger Arithmetic
({\justBAPA})
\cite{KuncakETAL06DecidingBooleanAlgebraPresburgerArithmetic}. The
decidability of this logic was long known
\cite{FefermanVaught59FirstOrderPropertiesProductsAlgebraicSystems},
but it was not until recently that
\cite{KuncakETAL06DecidingBooleanAlgebraPresburgerArithmetic}
proved that {\justBAPA} admits quantifier-elimination
and has asymptotically the same complexity as Presburger Arithmetic. 
The
quantifier-elimination algorithm introduced in \cite{KuncakETAL06DecidingBooleanAlgebraPresburgerArithmetic} reduces a
given {\justBAPA} formula to a Presburger arithmetic formula
using Venn regions. 
Many verification conditions expressing
properties of complex data structures can be
immediately formulated in quantifier-free fragment of {\justBAPA}
\cite{KuncakRinard07TowardsEfficientSatisfiabilityCheckingBoolean}, denoted {\BAPA}. 
For these theoretical and practical reasons, we consider only the {\BAPA}
fragment in this paper.

Checking the satisfiability of
a {\BAPA} formula is an NP-complete problem, where the non-trivial
aspect is showing the membership in NP
\cite{KuncakRinard07TowardsEfficientSatisfiabilityCheckingBoolean}.
The recent advances in SAT solvers made SAT instances coming from hardware
and software verification
more amenable to solution attempts than before. 
However, despite the existence of a polynomial encoding of {\BAPA} into SAT,
an efficient {\BAPA} solver is still
missing.
The most recent {\BAPA} implementation
\cite{SuterETALL2011BAPASMT} uses the state-of-art efficient
SMT solver Z3.
This implementation relies on the DPLL(T)
mechanism of Z3 to reason about the top-level propositional
atoms of a {\BAPA} formula. Although this implementation is
based on an an algorithm that explores all Venn
regions, it automatically decomposes problems into subcomponents when possible, and
applies Venn region construction only within individual components.
This approach is an important practical step forward, but there are still
natural formulas that cannot be decomposed. For such cases, the running time of the procedure increases
doubly-exponentially in the number of variables.\footnote{This paper was written in 2011; subsequently an \emph{implementation} 
  of another procedure is available in CVC4 and was documented in \cite{DBLP:journals/lmcs/BansalBRT18}, but this does not yield
  improved complexity bounds for the general case nor does it impact the problem studied in this paper.}

An alternative approach towards the efficient implementation is
to explore the sparse model property of {\BAPA}. In
\cite{KuncakRinard07TowardsEfficientSatisfiabilityCheckingBoolean}
was shown that, if a given {\BAPA} formula is satisfiable,
then there exists an equisatisfiable linear arithmetic
formula that is polynomial in the size of the original
formula. The decision procedure based on this theorem is
described in detail in
\cite{KuncakRinard07TowardsEfficientSatisfiabilityCheckingBoolean}. The procedure
takes as an input a {\BAPA} formula and converts it into an
equisatisfiable Presburger arithmetic formula
$\FPA$. Based on the newly derived $\FPA$ and the
theorem on a sparse solution for integer linear programming
\cite{EisenbrandShmonina06CaratheodoryBoundsIntegerCones},
the algorithm computes a positive integer $N'(d)$, which
denotes an upper bound of the number of non-empty Venn regions for formula that contains $d$ constraints.
The algorithm then runs in a loop and
tries to incrementally construct a model of sparsity $0, 1,
\ldots, N'(d)$. If no model is found after
the loop execution is finished, then the input formula is
unsatisfiable.
The number $N'(d)$ is an upper bound and it can be easily computed
from a dimension of a problem. However, this bound is not
tight. Our goal is to establish the bound on $N'(d)$ as tight as
possible in order to make an efficient
implementations of a {\BAPA} solver more feasible.

We are thus interested in deriving the smallest possible number
$N'(d)$, denoted by $N(d)$, which still preserves the desired property: if formula
has a solution, then it also has a solution of sparsity
$N(d)$. This paper will focus on computing the values of $N(d)$
using various combinatorial algorithms and their
optimizations.

The existence of $N(d)$ is guaranteed by the main theorem on
a sparse solution for integer linear programming
\cite{EisenbrandShmonina06CaratheodoryBoundsIntegerCones},
which states that if a vector is an element of an integer
cone, then it is also an element of some smaller integer
cone.  The key observation in
\cite{KuncakRinard07TowardsEfficientSatisfiabilityCheckingBoolean}
was not to use any ``small'' integer cone, but the smallest
one. For this purpose in
\cite{KuncakRinard07TowardsEfficientSatisfiabilityCheckingBoolean}
was introduced so called a {\emph{non-redundant integer
    cone}},
representing an integer cone that does not contain a
smaller cone that could generate a given vector.
This gives us a very simple definition of a function for which we known
linear lower bounds and $O(d\log(d))$ upper bounds.


The key contribution of this paper is a computation of the exact tight values of $N(d)$
for certain $d$. We also slightly improve previously known bounds for $N'(d)$. A computation of $N'(d)$ is an algorithmically
challenging task. Earlier computations \cite{KuncakRinard07TowardsEfficientSatisfiabilityCheckingBoolean}
found the exact values only for $d = 1, 2, 3$.

The following table outlines in comparison the information we knew about $N(d)$ earlier and the new values derived in this paper:
\[
\begin{array}{|c|c|c|c|c|} \hline 
   & \multicolumn{2}{|c|}{\mbox{previous known}} & \multicolumn{2}{|c|}{\mbox{new results}}      \\ \hline
d  & \mbox{lower bound} & \mbox{upper bound} & \mbox{lower bound} & \mbox{upper bound} \\ \hline 
1  &        1           &          1         &         1          &          1         \\ \hline 
2  &        2           &          2         &         2          &          2         \\ \hline 
3  &        3           &          3         &         3          &          3         \\ \hline
4  &        5           &          16        &         5          &     \mathbf{5}     \\ \hline
5  &        6           &          22        &        \mathbf{7}  &     \mathbf{7}     \\ \hline
6  &        7           &          29        &        \mathbf{9}  &     \mathbf{9}     \\ \hline 
7  &        8           &          36        &        \mathbf{11} &     \mathbf{19}    \\ \hline 
8  &        10          &          43        &        \mathbf{13} &          43        \\ \hline 
9  &        11          &          51        &        \mathbf{14} &          51        \\ \hline 
10 &        12          &          59        &        \mathbf{16} &     \mathbf{58}    \\ \hline 
\end{array}
\]

We obtained those results by applying the following techniques:
\begin{enumerate}
\item we reduced the search space by establishing the
  isomorphisms between solutions
\item we used the already known solution for $N(d)$ to obtain better estimate on $N'(d+1)$
\item we applied Gaussian elimination method to check whether the returned set of vectors is indeed a non-redundant integer
    cone
\item we used an implementation of a randomized algorithm to obtain a better estimation for some $N'(d)$ values, and
		for checking  efficiency of our implementation
\end{enumerate}


\section{Preliminaries}

This section summarizes the previously known results that
are necessary for a better understanding of the rest of the
paper. We recall the definitions and the theorems introduced
in
\cite{KuncakRinard07TowardsEfficientSatisfiabilityCheckingBoolean}.

Quantifier-free Boolean Algebra with Presburger Arithmetic
({\BAPA}) is a theory that includes reasoning about set
relations and operations, and reasoning about integer linear
arithmetic. Sets and integers are connected through the
cardinality operator. A simple decision procedure for
{\BAPA} uses Venn regions and reduces checking satisfiability
of a {\BAPA} formula to checking satisfiability of a
corresponding linear integer arithmetic formula. As an
illustration consider the following {\BAPA} formula:

\begin{equation*}
|U| = 100 \wedge \bigwedge_{1 \leq i < j \leq 3} {|x_i \cup
  x_j| = 30} \wedge \bigwedge_{1 \leq i \leq 3}{|x_i| = 20}
\wedge \bigwedge_{1 \leq i \leq 3}{|x_i| \subseteq U}.
\end{equation*}

With $l_i$ we denote fresh integer variables. The above
formula is equisatisfiable with the following formula written in a matrix form:
\begin{equation}\label{fla:derived}
\begin{bmatrix}
1 & 1 & 1 & 1 & 1 & 1 & 1 & 1\\
0 & 0 & 1 & 1 & 1 & 1 & 1 & 1\\
0 & 1 & 0 & 1 & 1 & 1 & 1 & 1\\
0 & 1 & 1 & 1 & 0 & 1 & 1 & 1\\
0 & 0 & 0 & 0 & 1 & 1 & 1 & 1\\
0 & 0 & 1 & 1 & 0 & 0 & 1 & 1\\
0 & 1 & 0 & 1 & 0 & 1 & 0 & 1
\end{bmatrix}
\begin{pmatrix}
l_{000}\\
l_{001}\\
l_{010}\\
l_{011}\\
l_{100}\\
l_{101}\\
l_{110}\\
l_{111}
\end{pmatrix}
 =
\begin{pmatrix}
100\\
30\\
30\\
30\\
20\\
20\\
20
\end{pmatrix}
\end{equation}

The details of the translation algorithm can be found in
\cite{KuncakRinard07TowardsEfficientSatisfiabilityCheckingBoolean}.
However, this newly derived formula might have an
exponential size in the size of the original formula. In
this example too, the number of variables is exponential
in the number of sets in the original formula.

\begin{definition}
Let $X \subseteq \mathbb{Z}^d$ be a set of integer vectors. An
integer cone generated by $X$, denoted with
$\rm{int\_cone}(X)$, is a linear additive closure of vectors
of $X$:
\[
\rm{int\_cone}(X) = \{\lambda_1 x_1 + \ldots + \lambda_n x_n
| n \geq 0, \ x_i \in X , \ \lambda_i \ge 0\ , \ \lambda_i \in \mathbb{Z}\}.
\]
\end{definition}

Note that checking satisfiability of \eqref{fla:derived}
reduces to checking whether a vector belongs to an integer
cone. The number of vectors in the integer cone can be
infinite and we are interested in deriving the ``small''
subset of them that would still generate the same initially
given vector. We apply the results obtained in the
operational research community on sparse solutions of
integer linear programming problems.

\begin{theorem}[Theorem 1 in 
\cite{EisenbrandShmonina06CaratheodoryBoundsIntegerCones}]\label{tm:fritz}
  Let $X \subseteq \mathbb{Z}^d$ be a finite set of integer
  vectors and $M_X = \max\{n |\, n = |x_{ij}|, \ x_{ij}
  \text{ is an ordinate of vector } x_i, \ x_i \in
  X\}$. Assume that $b \in \rm{int\_cone}(X)$. Then there
  exists a subset $\tilde{X} \subseteq X$ such that $b \in
  \rm{int\_cone}(\tilde{X})$ and $|\tilde{X}| \leq 2 d
  \log_2{(4 d M_X)}$.
\end{theorem}

This theorem establishes the bound on the number of vectors
of the cone needed to generate a given vector. Because $M_X
= 1$, the number of vectors in the ``smaller'' cone is
bounded by $2d\log_2 d + 4d$. In
\cite{KuncakRinard07TowardsEfficientSatisfiabilityCheckingBoolean}
it was observed that this bound can be reduced to
$2d\log_2 d$ by taking into account that the vectors are non-negative.

\begin{definition}
Let $X \subseteq \mathbb{Z}^d$ and let $b$ be an integer
vector. Set $X$ is called a non-redundant integer cone
generator for $b$, denoted by $\rm{NICG}(X, b)$, if:
\begin{itemize}
\item $b \in \rm{int\_cone}(X)$
\item for every $x \in X$ holds: $b \notin \rm{int\_cone}(X \backslash \{x\})$.
\end{itemize}
\end{definition}

Nevertheless, we want to avoid computing a non-redundant
integer cone generator for every given vector. The following
theorem proved in
\cite{KuncakRinard07TowardsEfficientSatisfiabilityCheckingBoolean},
shows that it is enough to consider only one particular
vector, namely $\Sigma X = \sum_{x \in X}{x}$. We define $\rm{NICG}(X)$ as
$\rm{NICG}(X, \Sigma X)$.

\begin{lema}\label{lema:NICG}
Let $X \subseteq \mathbb{Z}^d_{\geq 0}$ be a set of
non-negative integer vectors. The following two statements
are equivalent:
\begin{itemize}
\item there exists a non-negative integer vector $b$ such
  that $\rm{NICG}(X, b)$ holds
\item $\rm{NICG}(X)$ holds
\end{itemize}
\end{lema}

Our original motivation was to check satisfiability of
{\BAPA} formulas. The decision procedure can be outlined as
follows: we reduce satisfiability of the initial {\BAPA}
formula to check the membership in an integer cone, where
the generating vectors are bit vectors. Applying
Theorem~\ref{tm:fritz} results in the small model
property. Therefore, our new goal becomes to compute the
number $N(d)$ for a given dimension $d$. The number $N(d)$
sets an upper bound on the cone size: if a vector is a
member of an integer cone, then it is a member of a cone
generated with at most $N(d)$ vectors. To translate it back
to the {\BAPA} satisfiability problem: if a {\BAPA} formula
is satisfiable, then it also has a model where at most
polynomially many Venn regions are non-empty. The number of
non-empty Venn regions is determined using $N(d)$. The
decision procedures runs in the loop from 0 to $N(d)$ and
tries to incrementally construct a model of a size $0, 1,
\ldots, N(d)$.

Lemma~\ref{lema:NICG} justifies the following definition:

\begin{definition}
Let $d$ be a non-negative integer. With $N(d)$ we denote the
cardinality of a set $X$ such that $\rm{NICG}(X)$ holds and
for any set $Y$ of a greater size does not hold
$\rm{NICG}(Y)$
\[
N(d) = \max\{|X| \ | \ X \subseteq \{0, 1\}^d, \ \rm{NICG}(X)  \}
\]
\end{definition}

Lastly we provide is the summary on known lower and
upper bounds on the value of $N(d)$, as well as the computed
values for $N(d)$ for some $d$:

\begin{theorem}\label{tm:knownNd}
For a positive integer $d \ge 1$ and $N(d)$ the following holds:
\begin{enumerate}
\item $d \le N(d)$
\item $N(d) \le (1 + \varepsilon(d))(d\log_2 d)$, where
  $\varepsilon(d) \le 1$ and $\displaystyle\lim_{d
  \rightarrow \infty} \varepsilon(d) = 0$
\item $N(d) + 1 \le N(d + 1)$
\item $N(d) = d$, for $ d = 1, 2, 3$
\item $N(d) > d$ for $d \ge 4$
\end{enumerate}
\end{theorem}

In the rest of the paper we will describe the algorithms and
optimizations we used to compute $N(4), N(5)$ and $N(6)$. We
will also provide improved lower bounds on $N(7)$ and
$N(8)$.

	
\section{Core Techniques: N(4)=5, N(5)=7}\label{sec:core}

In this section we present methods that we initially used to compute values of $N(4)$ and $N(5)$. Figure~\ref{fig:isInIntCone} describes a simple algorithm that checks whether a set of vectors $X \subseteq \{0, 1\}^d$ is a non-redundant integer cone. 

\begin{figure}[h]		
\begin{codebox}
			\Procname{$\proc{NICG}(X)$}
			\zi \Comment Global variable that stores $\rm{NICG}$ property of $X$.
			\zi	$found \gets \const{false}$
			\zi \For each vector $x \in X$
			\zi	\Do
						$\proc{inIntConeTest}(X \backslash \{x\}, \sum{X})$
			\zi		\If $found \isequal \const{true}$
			\zi		\Then
							\Return $\const{false}$
						\End
					\End
			\zi	\Return $\const{true}$
			\end{codebox}
	
			\begin{codebox}
			\Procname{$\proc{inIntConeTest}(X, b)$}
			\zi	\If $b \isequal 0$
			\zi	\Then
						$found \gets \const{true}$
			\zi		\Return
					\End
					
			\zi \If $X \isequal \emptyset$
			\zi	\Then
						\Return
					\End
			
			\zi $newB \gets b$
			\zi $x \gets $ take any element from $X$
			\zi	\While \const{true}
			\zi	\Do
						$\proc{inIntConeTest}(X \backslash \{x\}, newB)$
			\zi		$newB \gets newB - x$
			\zi		\If $found \isequal \const{true}$ or $newB$ contains negative component
			\zi		\Then
							\Return
						\End
					\End
\end{codebox}
		  \caption{Program \textsc{NICG}: checks whether for a set of integer vector $X$ holds $\rm{NICG}(X)$}
		  \label{fig:isInIntCone}
		\end{figure}

A simple incremental algorithm for computing the value $N(d)$ works as follows: the 
algorithm starts with $n$ and constructs a set $X$ of the cardinality $n$, which has the property $\rm{NICG}(X)$. In the next iteration $n$ gets increased and the algorithm repeats the same steps.
As soon as the algorithm encounters the first $n$ for which it cannot construct a 
$\rm{NICG}(X)$ of the cardinality $n$, it stops and returns $N(d) = n - 1$.
The correctness of this algorithm is guaranteed by the following theorem, originally proved in \cite{KuncakRinard07TowardsEfficientSatisfiabilityCheckingBoolean}:

\begin{lema}\label{lema:monotonic}
If $\rm{NICG}(X)$ and $Y \subseteq X$, then $\rm{NICG}(Y)$.
\end{lema}

Using this approach we computed $N(5) = 7$ after approximately 3 hours.
\vspace{0.5cm}

\smartparagraph{Optimization: Binary Search.} Instead of incrementally constructing all the sets, we can apply Lemma~\ref{lema:monotonic} together with Theorem~\ref{tm:knownNd} to devise an algorithm that computes
the value of $N(d)$ in the binary search manner.
The algorithm makes a guess $n$ on the value $N(d)$ and tries to construct a set $X$ such that $|X| = n$ and $\rm{NICG}(X)$. If no such set exists, then $N(d) < n$, otherwise $n \le N(d)$.

As an illustration, consider $d = 5$. Applying Theorem~\ref{tm:knownNd} to compute the bounds on the value of $N(5)$, 
the algorithm derives the interval in which $N(5)$ occurs: $6 \le N(5) \le 11$. The first guess is $N(5) = 8$.
Then the algorithm tries to construct a set $X$ such that $|X| = 8$ and $\rm{NICG}(X)$. Because such a set does not
exist, the algorithm will not construct it implying $6 \le N(5) \le 7$. The next guess is $N(5) = 7$.
Since there exists a set $X$ such that $|X| = 7$ and $\rm{NICG}(X)$, the algorithm will construct such a set and output 
$N(5) = 7$.

\smartparagraph{Incremental Construction vs Binary Search.} We have implemented both,
the incremental construction and the binary search approach, to derive $N(5)$.
The binary search approach found $N(5)$ faster than the incremental construction approach. However, our experimental results
show that the binary search approach is slower than the incremental construction approach in computing $N(d)$ for $d > 5$.
The difference in the experimental results is caused by the fact that testing the existence of a set $X$ such that $\rm{NICG}(X)$ is computationally more expensive than testing the existence of a set $Y$ such that $\rm{NICG}(Y)$ when
$|X| > |Y|$. Another issue with the binary search approach is that if the initial interval
is not tight enough, the algorithm might make a guess on $N(d)$ that is significantly larger than the value $N(d)$ itself.

As an example consider $d = 5$. In the incremental approach the algorithm must examine at most $\binom{31}{6} + \binom{31}{7} + \binom{31}{8} = 11254581$ sets of vectors. In the binary search approach the algorithm must examine $\binom{31}{7} +  \binom{31}{8} = 2921750$ sets of vectors, where the value 31 represents cardinality
of the set $\{0, 1\}^5 \setminus \{0, 0, 0, 0, 0\}$.

\smartparagraph{Optimization: Preserving Sums.} In order to obtain a more efficient computation of $N(d)$ we tried
		an approach based on preserving sums of vectors, which can be later reused in the computation.
		The idea on preserving sums was motivated by the following observation:
		if $Y \subset X$ and $\rm{NICG}(X)$, then $\Sigma X \not\in \rm{int\_cone}(Y)$. 
		To benefit from the observation, for every examined $Y$ for which $\rm{NICG}(Y)$ holds the
		algorithm must keep track of the sum $\Sigma Y$.
		We have tried this heuristic, but did not obtain any significant improvement.
		The advantage of such an approach is that 
		the algorithm can compute new sums quickly, and detect not $\rm{NCIG}$ faster.
		The disadvantage
		is the process of maintaining sums. The search algorithm must be aware which sums should be stored and which removed.
		In certain cases the search algorithm must
		copy the whole data structure that keeps the sums. Our experiments have shown that maintaining so much
		information is more costly than the calculation.
		
	\subsection{Isomorphic sets}\label{subsec:isomorphic}
		So far the algorithms searched for the solution over \textbf{all} sets of vectors of given cardinality. We applied
		optimizations to early detect if set does not improve the solution. Also, we introduced
		approaches which improved maintaining information about the sets. But common to all those cases was that almost all the sets were examined. This was a big drawback: we will demonstrate that one does not need to examine all the sets.
To motivate our observations, consider the following two sets: $X_1 = \{(1, 1, 0), (0, 1, 0)\}$ and $X_2 = \{(0, 1, 1), (0, 1, 0)\}$.
		Performing a permutation $\begin{pmatrix} 1 & 2 & 3 \\ 3 & 2 & 1 \end{pmatrix}$ on indices of components of the 
		vectors in $X_1$ we obtain $X_2$. Permuting components of the vectors does not affect the solution.
		Therefore, if set $X_1$ does not lead to the solution, then $X_2$ does not lead, too. Similarly, if
		$X_1$ leads to the solution, $X_2$ leads as well. The observation allow us to consider only vectors
		that are not isomorphic, where isomorphism between two sets of vectors is defined as follows:
		\begin{definition}
			We say that two sets of vectors $X, Y \subseteq \{0, 1\}^d$ are isomorphic if there exists 
			a permutation $P$ over the set $\{1, \ldots, d\}$
			and a bijective function $f_P : X \rightarrow Y$ defined as
			\[
				f_P(x) = y \Leftrightarrow x_i = y_{P(i)}, i = 1, \ldots, d.
			\]
		\end{definition}
		Basically, there are two ways to check, call it check functions, whether we already considered an isomorphic set:
		\begin{enumerate}
			\item For each considered set $X$ so far, mark all isomorphic sets to $X$, i.e. mark all $d!$ sets (note that some
						of them might repeat) storing them in a structure $marked$. Before a new set is processed
						check whether it is in $marked$.
			\item Store each considered set in a structure $done$. When there is a new set $X$ to be examined, run all
						$d!$ permutations on $X$. For each permutation $p$ check if $p(X)$ is in $done$.
		\end{enumerate}
		We used this approach for $d \leq 7$. There are $2^{64}$ different sets in case $d = 6$. A particular
		set can be isomorphic to at most $6!$ other sets. Because this is an equivalence relation, at least $\frac{2^{64}}{6!}$
		non-isomorphic sets should be stored somehow. This is far away too much. To avoid this problem, we 
		can use a bit different method. Let us define 
		\[
			X^{(k)} = \{x | x \in X \rm{\ and\ } x \rm{\ contains\ exactly\ } k \rm{\ non-zero\ components}\}.
		\]
		Then we say $X$ and $Y$ are isomorphic if $(X^{(1)}, X^{(2)})$ is isomorphic to $(Y^{(1)}, Y^{(2)})$. 
		With such a method, sets $X = \{(1, 0, 0, 0)$ $, (1, 1, 1, 0)\}$
		and $Y = \{(1, 0, 0, 0), $ $(1, 1, 0, 1)\}$ will not be considered as isomorphic, although they are isomorphic
		by a permutation $\begin{pmatrix} 1 & 2 & 3 & 4 \\ 1 & 2 & 4 & 3 \end{pmatrix}$. Thus, it does not cover
		all isomorphic pairs, but allow us to shrink the usage of memory.
		
		Finally, the required memory is sufficiently small that we can
		store sets in the both cases in an array. As a result, the check 
		function can be executed in a constant time.
		The first approach uses $O(d! \cdot T)$ time ($T$ is number of non-isomorphic sets) for marking, 
		and $O(d! \cdot M)$ memory. Once we do that, 
		the check function is performed in a constant time.\\
		The second approach uses $O(1)$ time to store an examined set, and it uses $O(M)$ memory.
		For each stored set $X$ there are multiple isomorphic sets. Note that some of those sets do not 
		always have $d!$ isomorphic
		sets, like for example $\{(1, 0, 0, 0)\}$. Actually, most of the time they have less than $d!$.
		Before storing a set $X$ we have to generate $d!$ other sets and check whether they are in $done$ or not.
		There are $T'$ isomorphic sets, where for $d = 6$ the value $T'$ is a few hundred times bigger than $T$.
		This approach gives the time complexity $O((T + T') \cdot d!)$, and the memory complexity $O(M)$.\\
		In our case, we have already shrank memory, thus the memory is not an issue, but time efficiency.
		Therefore we choose the first approach.
		
		Using this optimization we obtained a method that in a few minutes calculates $N(5) = 7$.

\section{Gaussian Elimination: N(6)=9}

	In Section~\ref{sec:core} we described a different approaches towards a construction of  $\rm{NICG}$ sets. Most of the approaches consider $\rm{NICG}$ property only of
	currently calculated set. We also argued why the approach in which we try to maintain as many information as possible is not very efficient. Knowing that, we decided
	to merge ideas from the both approaches and come with a more efficient algorithm.\\
	The idea about maintaining too much information is not good, as we have explained. However, maintaining some amount
	of ``not $\rm{NICG}(Y)$'' information might reduce the search space, and thus improve the running time. 
	The property ``not $\rm{NICG}(Y)$'' allow us to avoid computing over every set $X$ such that $Y \subset X$.
	For instance, consider an example where $Y = \{(1, 0, 0, 0, 0), (0, 1, 0, 0, 0), $ $(1, 1, 0, 0, 0)\}$.
	$Y$ is not $\rm{NICG}$ because $\sum{Y} = 2 (1, 1, 0, 0, 0)$. There are $\binom{28}{4} + \binom{28}{5} +
	\binom{28}{6} + \binom{28}{7} + \binom{28}{8} = 4787640$ sets $X$, such that $|X| \leq 8, Y \subset X$. Most of
	those sets would be considered if we do not shrink them using the fact not $\rm{NICG}(Y)$. Quite a lot
	of sets is affected by only one set, thus we decided to use information about not $\rm{NICG}$ sets.\\
	On the other side, we decided not to use information about $\rm{NICG}$ sets, but for a given set $X$ to test
	$\rm{NICG}(X)$ on fly. In Section \ref{sec:introduction} we saw that answering
	$\rm{NICG}(X)$ is the same as answering is there a solution to the corresponding system of equations. Instead of 
	using the procedure \proc{inIntConeTest} to answer that, we try to solve system using Gaussian elimination.
	
	Answering whether a given set $X$ has the $\rm{NICG}(X)$ property by solving the corresponding
	system with Gaussian elimination might look like an inefficient approach. To understand such a view, consider
	a system of five equations and eight variables (what could be the case for $d = 5$) such that its
	solution contains three parameter-variables. 
	Each component of a vector sum of the eight binary vectors is a non-negative integer value not greater than 8.
	Therefore there are $9^3$ possibilities to assign values to the three parameters.\\
	A system that represents a set for $d \ge 6$ might contain even more parameter-variables
	resulting in even more possible assignments to the parameter-variables.
	However, it turned out that, for the systems our search algorithms constructed, Gaussian method works very good since
	most of the parameter-variable assignments are not valid.
	
	This approach verified result for $d \leq 5$ and gave $N(6) = 9$ in around a thirty minutes.\\
	Below are given examples of sets of vectors that represent solutions for $d = 1 \ldots 6$. Those sets
	are obtained by applying the described approaches.
	
	\begin{table}
		\begin{center}
			\begin{tabular}{ c || c | c | c | c | c | c}
		  $d$ & 1 & 2 & 3 & 4 & 5 & 6 \\ \hline
		  $N(d)$ & 1 & 2 & 3 & 5 & 7 & 9 \\ \hline
		  a solution & $\begin{pmatrix} 1 \end{pmatrix}$
		  & $\begin{pmatrix} 1 & 1 \\ 0 & 1 \end{pmatrix}$
		  & $\begin{pmatrix} 1 & 1 & 1 \\ 1 & 0 & 1 \\ 0 & 1 & 1 \end{pmatrix}$
		  & $\begin{pmatrix} 1 & 1 & 1 & 1 & 0 \\ 1 & 1 & 1 & 0 & 1 \\ 0 & 1 & 0 & 1 & 1 \\ 0 & 0 & 1 & 1 & 1\end{pmatrix}$
		  & $\begin{pmatrix} 1 & 1 & 1 & 1 & 0 & 0 & 1 \\ 1 & 1 & 0 & 0 & 1 & 1 & 1 \\
		  										 0 & 1 & 1 & 1 & 1 & 0 & 0 \\ 0 & 0 & 1 & 0 & 0 & 1 & 1 \\
		  										 0 & 0 & 0 & 1 & 1 & 1 & 1 
		  		 \end{pmatrix}$
		  & $\begin{pmatrix} 1 & 1 & 1 & 1 & 0 & 0 & 1 & 1 & 1 \\ 1 & 1 & 1 & 0 & 1 & 1 & 0 & 1 & 1 \\
		  										 0 & 1 & 0 & 1 & 1 & 0 & 0 & 0 & 1 \\ 0 & 0 & 1 & 1 & 0 & 1 & 0 & 1 & 0 \\ 
		  										 0 & 0 & 0 & 0 & 1 & 0 & 1 & 1 & 0 \\ 0 & 0 & 0 & 0 & 0 & 1 & 1 & 0 & 1 
		  		 \end{pmatrix}$
			\end{tabular}
			\caption{Solutions for different $d$ values, an example solution per a value. Full set of solutions is available at \cite{bapasite}.}
		\end{center}
	\end{table}
	
\section{Speeding up Search using Weak Isomorphisms}

	In Subsection \ref{subsec:isomorphic} we have seen two approaches that might be used to eliminate isomorphic states.
	One of the approaches is time inefficient, and another one is memory costly. Both, time and memory inefficiency,
	grows exponentially, which suggests that any of those approaches can be used only for small $d$ values.	On the
	other side, both of these methods are very strict in sense that for a given type of isomorphism, the methods
	eliminate all isomorphic states (i.e. detect all isomorphic $(X^{(1)}, X^{(2)})$ states, as has been already explained).
	
	Every approach to the problem we have used so far can be described by the following algorithm:
		\begin{codebox}
		\zi \Comment $currSolution$ represents a NICG set of vectors that is
		\zi		already considered as part of some solution.
		\zi \Comment $nonUsedVectors$ represents set of vectors that 
		\zi		can be used to build a solution from the current $state$.
		\Procname{$\proc{Solve}(currSolution, nonUsedVectors)$}
		\zi	update $solutions$ by $currSolution$ \Comment $solution$ is a global set of solutions.
		\zi	$vector \gets$ choose element from $nonUsedVectors$
		\zi	$nonUsedVectors \gets$ $nonUsedVectors \backslash vector$
		\zi	update isomorphic states
		\zi	\If $\proc{NICG}(currSolution \cup vector)$
		\zi	\Then
					$\proc{Solve}(currSolution \cup vector, nonUsedVectors)$
				\End
		\zi	$\proc{Solve}(currSolution, nonUsedVectors)$
		\end{codebox}
		\label{prog:Solve}
	
	The algorithm above generates a search tree. The method we use to eliminate isomorphic states in the search tree
	directly affects both, the running time and the memory usage.\\
	As we can see, on one side are the introduced methods that eliminate a lot of states, but they
	use too much time, or too much memory. On the other side, if we do not use any method for elimination 
	we have to search over a huge tree, but then we do not use any extra time or memory for elimination. Instead of devising
	a method that rely on benefits only of one or the another side, we have tried to ``meet in the middle''.\\
	A method that we describe is not so strict in elimination, as the previous methods were, but it is very
	efficient as we are going to show by the results.
	
	Suppose the algorithm is in a state $(currSolution = X, nonUsedVectors = A)$, and it chooses to examine
	$vector \gets x$. In this
	state the algorithm must decide which states $(X \cup y, A \backslash x)$ it is \textbf{not} going to visit, knowing
	that it is going to visit $(X \cup x, A \backslash x)$. Note that even if $\proc{NICG}(X \cup x)$ returns
	$\const{false}$, the state $(X \cup x, A \backslash x)$ can be considered as a visited one, but such that it
	does not lead to the solution. Obviously, if $X \cup x$ is isomorphic to $X \cup y$, there is no
	point to search over subtree represented by $X \cup y$.
	By isomorphism between sets $X$ and $Y$ we denoted existence of a permutation
	that maps $X$ to $Y$. Additionally, we observe that
	if there exists a permutation $P$ such that $P(X) = X$ and $P(x) = y$,
	then there exists a permutation $P'$ such that $P'(X \cup x) = X \cup y$. The opposite does not
	stand always.\\
	Consider even more specific type of permutations that make two states being isomorphic. We say
	that a permutation $P$ 'preserves order of ones' of a collection of vectors $X$, if the following holds:
	\[
		(\forall x \in X) (\forall i \in \{1, \ldots, d\}) x_i = 1 \Rightarrow P(i) = i,
	\]
	where by $x_i$ is denoted vector-component of $x$ at the position $i$.
	In other words, when $P$ is applied on $X$ it does not
	change order of ones in that collection of binary vectors. Of course, it immediately leads to the conclusion
	$P(X) = X$. This type of permutations we call ``1-order preserving'' permutations.
	
	In our algorithm we describe 1-order preserving permutations by using
	a single boolean array $fixedPerms$ of the size $d$. If $fixedPerms[i] = \const{true}$ it means that
	the array represents a collection of permutation such that for every permutation $P$ from the collection 
	it stands $P(i) = i$. The array $fixedPerms$ is updated in the following way:
	\begin{itemize}
		\item Initially, $fixedPerms = \{\const{false}\}^d$.
		\item When a new vector $x$ is added to the current state $X$, array $fixedPerms$ is updated
					as follows:
					\begin{codebox}
					\zi	\For $i \gets 1 \ldots d$
					\zi	\Do
								\If the $i$-th component of $x \isequal 1$
					\zi		\Then
									$fixedPerms'[i] \gets \const{true}$
								\End
							\End
					\end{codebox}
	\end{itemize}
	Therefore, for each newly added vector the algorithm updates $fixedPerms$ in $O(d)$ time.
	For each state the algorithm needs additional
	$d$ bits to represent the corresponding array.
	Testing whether two vectors $x$ and $y$ are isomorphic
	with respect to 1-order preserving collection given by $fixedPerms$ can be done as follows:
		\begin{codebox}
		\Procname{$\proc{IsomorphicVectors}(x, y, fixedPerms)$}
		\zi	\If number of ones in $x \neq $ number of ones in $y$
		\zi	\Then
					\Return \const{false}
				\End
		\zi	\For $i \gets 1 \ldots d$
		\zi	\Do
					\If $fixedPerms[i] \isequal \const{true}$
		\zi		\Then
						\If the $i$-th components of $x$ and $y$ differ
		\zi			\Then
							\Return \const{false}
						\End
					\End
				\End
		\zi	\Return \const{true}
		\end{codebox}
	The method $\proc{IsomorphicVectors}$ for particular input executes in $O(d)$ time.
	
	As we have seen, eliminating states according to 1-order preserving collection is both time and memory
	efficient. However, such approach is a bit weak, and it does not eliminate all isomorphic states. This can
	be illustrated with the following example. Suppose $currSolution = \{(1, 1)\}$ and
	$nonUsedVectors = \{(1, 0), (0, 1)\}$. The only permutation that does not change order of ones in 
	$currSolution$ is $\begin{pmatrix} 1 & 2 \\ 1 & 2 \end{pmatrix}$, therefore $fixedPerms = \{\const{true}, \const{true}\}$.
	Thus, call of $\proc{IsomorphicVectors}((1, 0), (0, 1), \{\const{true}, \const{true}\})$ will
	return $\const{false}$. However, there exists a permutation $\begin{pmatrix} 1 & 2 \\ 2 & 1 \end{pmatrix}$
	that maps $\{(1, 1), (1, 0)\}$ onto $\{(1, 1), (0, 1)\}$, what means that these two states should
	be considered as isomorphic.
	
	Although the last approach is based on weak isomorphism it runs faster then the previously described approaches
	for $d = 1 \ldots 6$.
	For $d = 6$ there exist 254 non-isomorphic solutions.
	The algorithm that uses Gaussian elimination finds 80000 solutions, with many of them being isomorphic. 
	The search lasts for about 45 minutes.
	The algorithm that implements the weak isomorphism finds nearly 5000 solutions in about 5 minutes.

\section{Randomizing Search Order; New Lower Bounds}

	The algorithms we have presented so far deterministically choose the next state from the current state.
	We developed a randomized algorithm that chooses a next state randomly in such a way that all states
	reachable from the current state have the same probability to be chosen.
	The uniform probability distribution on choice over the states is achieved by shuffling list of
	states reachable from the current state, and then picking the first state from the shuffled list as the next state. 
	Although the implementation difference is minor,
	results are significantly better than using deterministic algorithm, as can be seen in the following paragraph.
	
	We did not succeed to get exact values for $N(7)$, $N(8)$, $N(9)$ or $N(10)$. Instead, we got better lower estimate
	of these values. Those estimates are: $N(7) \geq 11$; $N(8) \geq 13$; $N(9) \geq 14$; $N(10) \geq 16$.
	It is interesting that in less than a second we got 
	result $N(6) \geq 9$. In a few minutes we got $N(7) \geq 11$, and in an hour we got $N(8) \geq 13$, $N(9) \geq 14$ and
	$N(10) \geq 16$. On the other hand, deterministic algorithm uses the following
	amount of time: $N(6) \geq 9$ in a minute; $N(7) \geq 11$ after a few hours; 
	$N(8) \geq 13$ we did not get even after a day of running the algorithm.
	
\section{Better estimate of $N'(7)$ using Decomposition: from 36 to 19}

	The best known estimate so far for $N'(7)$ is 36. We successfully improved this upper bound to 19 as follows:\\
	Let $X$ be a solution for $d = 7$, i.e. $X \subseteq \{0, 1\}^7$ and $|X| = N(7)$. 
	Then set $X$ can be decomposed into two subsets $X_0$ and $X_1$ such that
	\begin{itemize}
		\item $X_0 \cap X_1 = \emptyset$,
		\item $X_0 \cup X_1 = X$,
		\item $X_0$ contains only vectors which first component is 0,
		\item $X_1$ contains only vectors which first component is 1.
	\end{itemize}
	From Lemma \ref{lema:monotonic} it follows that $X_0$ and $X_1$ are NICG sets. Since the first
	component of each vector in $X_0$ is 0, $X_0$ can be considered as
	a set of 6-dimensional vectors. Therefore, $|X_0| \le N(6)$.
	In order to estimate an upper bound on $|X_1|$, we use the same algorithm as we use to obtain $N(6)$,
	with the input defined as set $\{x\ |\ x \in \{0, 1\}^7 \wedge x_1 = 1\}$.
	Running the algorithm on such a set we get the final result, a set $Y$, after 30 minutes with
	$|Y| = 10$. Therefore, $|X_1| \leq 10$. Since every solution for $d = 7$ can be decomposed into $X_0$
	and $X_1$, with upper bounds 9 and 10, respectively, it implies that solution for $d = 7$ is of cardinality of at
	most $9 + 10 = 19$.

\section{Improvement of $N'(d)$ for Arbitrary $d$}\label{section:row-isomorphism}

\subsection{Isomorphism on row additions}\label{subsection:row-isomorphism}
	Each NICG set of vectors $X$ can be described as a matrix of dimension $d \times |X|$, where each column
	of the matrix represents a single vector from $X$, and no two different columns represent the same vector.
	In this section we introduce an isomorphism of NICG solutions that involves additions and
	substructions on rows of matrices.
	
	Consider an NICG set of vectors $X$, and the corresponding matrix $M$. We will say that two rows, $i_1$ and $i_2$,
	do not share variable if there does not exists $j$ such that $M_{i_1, j} = 1$ and $M_{i_2, j} = 1$.
	We can state the following lemma.
	\begin{lema}\label{lema:row-addition}
		Consider a matrix $M$, and assume that $M$ contains at least two rows that do not share variable. 
		Let two of these rows be $i_1$ and $i_2$. Let a matrix $M'$ be obtained
		from the matrix $M$ by replacing the row $i_2$ by the row-sum $i_1 + i_2$. $M$ represents NICG set of vectors
		if and only if $M'$ represents NICG set of vectors.
	\end{lema}
	\begin{corollary}\label{cor:subset-row}
		For NICG set $X$, and its corresponding matrix $M$, such that $M$ contains two rows $i_1$ and $i_2$ such that
		every variable presented in $i_1$ is presented in $i_2$ as well, there exists an NICG set $X'$ which can be obtained
		from $X$ by replacing the row $i_2$ by the row-subtruction $i_2 - i_1$.
	\end{corollary}
	Lemma \ref{lema:row-addition} gives a new insight about isomorphic NICG sets. Using the lemma,
	we give a better estimate of $N'(d)$, as presented
	in Section \ref{section:upper-bound-improvement}.

\subsection{Upper Bound Improvement for large $d$}\label{section:upper-bound-improvement}
	In this section we give a better estimate of $N'(d)$. The improvement rely on result presented in
	Section \ref{subsection:row-isomorphism}. State the following lemma:
	\begin{lema}\label{lema:upper-bound}
		For each NICG set $X$, and its corresponding matrix $M$, there exists an NICG set $X'$, 
		along with its corresponding matrix $M'$, such that 
		\begin{enumerate}[(1)]
			\item $|X| = |X'|$, and 
			\item every row in $M'$ contains at least one value 0.
		\end{enumerate}
	\end{lema}
	\begin{proof}
		If $M$ satisfies the condition (2), then let $X' = X$ and the proof is done. Therefore, assume
		that $M$ contains a row $i$ such that all its values are 1. It means that variables of every row $j$ 
		are contained in row $i$ as well. By Corollary \ref{cor:subset-row}, there exists $X_1$ that is obtained from $X$
		by replacing the row $i$ by the row $i - j$. If $X_1$ contains a row with 1s only, then we are going to apply
		Corollary \ref{cor:subset-row} on $X_1$ getting $X_2$. We continue this process until we get $X_r$ that
		does not contain row with all 1s. Because there is finite number of rows, $X_r$ will be obtained in a finite
		number of steps. Once we obtain $X_r$, let $X' = X_r$. By the construction and the corollary, $X_r$
		satisfies both (1) and (2), implying $X_r$ is NCIG.
		
		This completes the proof.
	\end{proof}
	Acording to Lemma \ref{lema:upper-bound}, there exists a solution $X$ to the problem such that the corresponding matrix 
	does not contain a row with all values being 1. 
	Thus we have that every component of the sum of vectors of $X$ never exceed $N - 1$.
	If we recall proof of Theorem 2 in
	\cite{KuncakRinard07TowardsEfficientSatisfiabilityCheckingBoolean}, we obtain an upper bound on $N$ to be 
	the maximal value such that 
	\[
		2^N \leq N^d.
	\]
	The last inequality gives a slight improvement on an upper bound of $N(d)$.
	
	The upper bound can be even more improved by using result $N(d) > d + 1$, for $d > 4$, given in
	\cite{KuncakRinard07TowardsEfficientSatisfiabilityCheckingBoolean}.
	Consider a solution $X$ and its corresponding matrix $M$ for $d > 4$. By Lemma \ref{lema:upper-bound}
	there exists a solution such that each row contains at least one 0 value. Assume each of the rows contains 
	a single 0 value.
	Therefore, at least $N(d) - d$ columns would contain 1s only. Since $N(d) - d \ge 2$ for $d > 4$,
	the system is not NICG. Therefore, the assumption is wrong
	and there exist at least one row that contains 2 zeros. The last gives a new improvement on 
	upper bound on $N$ that can be described by the following inequality:
	\[
		2^N \leq N^{d - 1} \cdot (N - 1).
	\]


\section{Different Approaches}

	In order to compute $N(d)$ we used general solvers for solving systems of equations with non-negative 
	integer variables. The solvers were used
	to give an answer whether a particular set of vectors $X$ has the $\rm{NICG}$ property.
	If $M$ is the matrix that represents $X$, then existance of the
	NICG property can be answered as follows:
	\begin{itemize}
		\item Let $M^k$ be defined as a matrix which $k$-th column contains only 0s, and every other column
					is a copy of the corresponding column in $M$.
		\item If for every matrix $M^k$, for $1 \le k \le \rm{\# of\ columns \ in \ } M$, 
					there does not exist a non-negative integer solution, 
					then $X$ has the $\rm{NICG}$ property, otherwise it does not have.
	\end{itemize}
	
	We tried this approach using standard solvers GLPK and jOpt. None of them was nearly as efficient on this problem as our 
	implementation of NICG property using Gaussian elimination. We suspect that the reason behind
	this is the fact that we were mainly working with systems of a small number of equations and that our internal implementation avoids calls into external libraries.
        Furthermore, we faced memory management bugs in GLPK.

\section{Conclusions}

{\BAPA} has the small model property. We are interested in deriving a number $N(d)$,
where $N(d)$ is the smallest number such that the following property holds: if formula has a solution, then
it also has a solution of the size $N(d)$.
In this paper we computed the values $N(4)$, $N(5)$ and $N(6)$. We also significantly improved the known bounds for $N(7)$, $N(8)$, $N(9)$ and $N(10)$. Those numbers are used to determine the size of small models for {\BAPA} formulas.

Our motivation was twofold: first, we obtained the bounds that
improve constant factors in asymptotically best 
algorithm for {\BAPA}. Second, we provided another case study 
in developing domain-specific algorithms for combinatorial
search. Although we found a domain-specific search algorithm
to be the most effective, the problem may prove to be fruitful ground
for future general-purpose constraint-solving techniques, such as
pseudo-Boolean and finite-domain solvers.

\bibliographystyle{abbrv}
\bibliography{pnew}
\end{document}